%% file: main.tex
\documentclass[letterpaper, 10 pt, conference]{ieeeconf}  

\IEEEoverridecommandlockouts                              
\def\BibTeX{{\rm B\kern-.05em{\sc i\kern-.025em b}\kern-.08em
    T\kern-.1667em\lower.7ex\hbox{E}\kern-.125emX}}

\input{ancillary/packages}

\usepackage{amsthm} 
\input{ancillary/macros}
\input{ancillary/colors}

\newcommand{\change}[1]{#1} 

\begin{document}

\title{
\change{
A Stability-Based Abstraction Framework for Reach-Avoid Control \\ of Stochastic Dynamical Systems with Unknown Noise Distributions
}
\thanks{
This research has been partially funded by an NWO grant NWA.1160.18.238 (PrimaVera), ERC Starting Grant 101077178 (DEUCE), and a JPMC faculty research award.
T. Badings and N. Jansen are with the Institute for Computing and Information Sciences, Radboud University, Nijmegen, The Netherlands; \{{\tt\small thom.badings, nils.jansen}\}{\tt\small @ru.nl}.
N. Jansen is also with the Faculty of Computer Science, Ruhr-Universität Bochum, Germany.
L. Romao and A. Abate are with the Department of Computer Science, Oxford University; \{{\tt\small licio.romao, alessandro.abate}\}{\tt\small @cs.ox.ac.uk}.
} 
}

\author{Thom Badings, Licio Romao, Alessandro Abate, Nils Jansen}

\maketitle

\begin{abstract} 
Finite-state abstractions are widely studied for the automated synthesis of correct-by-construction controllers for stochastic dynamical systems. However, existing abstraction methods often lead to prohibitively large finite-state models. To address this issue, we propose a novel abstraction scheme for stochastic linear systems that exploits the system's stability to obtain significantly smaller abstract models. As a unique feature, we first stabilize the open-loop dynamics using a linear feedback gain. We then use a model-based approach to abstract a known part of the stabilized dynamics while using a data-driven method to account for the stochastic uncertainty. We formalize abstractions as Markov decision processes (MDPs) with intervals of transition probabilities. By stabilizing the dynamics, we can further constrain the control input modeled in the abstraction, which leads to smaller abstract models while retaining the correctness of controllers. Moreover, when the stabilizing feedback controller is aligned with the property of interest, then a good trade-off is achieved between the reduction in the abstraction size and the performance loss. The experiments show that our approach can reduce the size of the graph of abstractions by up to 90\% with negligible performance loss. 
\end{abstract}

\input{sections/1-Introduction}
\input{sections/2-Background}
\input{sections/3-Abstraction}
\input{sections/4-Stability}
\input{sections/5-Experiments}
\input{sections/6-Conclusion}

\bibliographystyle{ieeetr}
\bibliography{literature,ECC-licio}

\end{document}

%% file: ancillary/packages.tex
\usepackage{graphicx}
\usepackage{amsmath,amssymb}       

\usepackage{hyperref}
\usepackage{cleveref}
\usepackage[noadjust,nospace]{cite}

\usepackage{nicefrac}
\usepackage{mleftright}
\usepackage{cases}
\usepackage{empheq}
\usepackage{mathbbol}
\usepackage{stmaryrd}
\usepackage{bm}
\usepackage{bbm}
\DeclareSymbolFontAlphabet{\amsmathbb}{AMSb}
\usepackage{mathtools}
\usepackage{accents}
\DeclareSymbolFontAlphabet{\amsmathbb}{AMSb}

\usepackage{listings}
\usepackage[colorinlistoftodos]{todonotes}

\usepackage{siunitx}
\usepackage{caption}
\usepackage{subcaption}
\usepackage{tabularx}
\usepackage{booktabs}
\usepackage{scrextend}
\usepackage{threeparttable}
\usepackage{marvosym}
\usepackage{multirow}
\usepackage{multicol}

\usepackage{xspace}
\usepackage{microtype}

\usepackage{cleveref}

\usepackage{tikz}
\usepackage{pgfplots}
\DeclareUnicodeCharacter{2212}{−}
\usepgfplotslibrary{external,colorbrewer}
\usetikzlibrary{patterns,shapes.arrows}
\pgfplotsset{compat=newest}

\definecolor{red}{rgb}{0.745,0.192,0.102}
\definecolor{darkgreen}{RGB}{34,161,55}
\definecolor{ruhuisstijlrood}{rgb}{0.745,0.192,0.102}
\definecolor{ruhuisstijlzwart}{rgb}{0,0,0}
\definecolor{ruhuisstijlwit}{rgb}{0.98,0.98,0.98}
\newcommand{\rured}[1]{\textcolor{ruhuisstijlrood}{#1}}

\definecolor{plotblue}{rgb}{0.1,0.498039215686275,0.9549019607843137}

\newcolumntype{H}{>{\setbox0=\hbox\bgroup}c<{\egroup}@{}}
\usepackage{colortbl}
\definecolor{LightCyan}{rgb}{0.88,1,1}

\newcommand*{\belowrulesepcolor}[1]{%
  \noalign{%
    \kern-\belowrulesep 
    \begingroup 
      \color{#1}%
      \hrule height\belowrulesep 
    \endgroup 
  }%
} 
\newcommand*{\aboverulesepcolor}[1]{%
  \noalign{%
    \begingroup 
      \color{#1}%
      \hrule height\aboverulesep 
    \endgroup 
    \kern-\aboverulesep 
  }%
} 

\usepackage{mdframed}

\usetikzlibrary{backgrounds,patterns,patterns.meta,arrows,arrows.meta,calc,shapes,shadows,decorations.pathmorphing,decorations.pathreplacing,automata,shapes.multipart,positioning,shapes.geometric,fit,circuits,trees,shapes.gates.logic.US,fit, shadows.blur,shapes.symbols,intersections}

%% file: ancillary/macros.tex
\newcommand{\prism}{\textrm{PRISM}\xspace}

\newtheorem{definition}{Definition}
\newtheorem{assumption}{Assumption}

\newtheorem{theorem}{Theorem}
\newtheorem{remark}{Remark}
\newtheorem{lemma}{Lemma}
\newtheorem{problem}{Problem}

\crefname{equation}{Eq.}{Eqs.}
\crefname{pluralequation}{Eqs.}{Eqs.}

\crefname{algorithm}{Algorithm}{Algorithm}

\crefname{figure}{Fig.}{Figs.}
\crefname{pluralfigure}{Figs.}{Figs.}

\crefname{section}{Sect.}{Sects.}
\crefname{pluralsection}{Sects.}{Sects.}

\crefname{table}{Table}{Table}
\crefname{pluraltable}{Tables}{Tables}

\crefname{definition}{Def.}{Def.}
\crefname{pluraldefinition}{Defs.}{Defs.}

\crefname{theorem}{Theorem}{Theorems}
\crefname{pluraltheorem}{Theorems}{Theorems}

\crefname{lemma}{Lemma}{Lemmas}
\crefname{plurallemma}{Lemmas}{Lemmas}

\crefname{example}{Example}{Example}
\crefname{pluralexample}{Examples}{Examples}

\crefname{problem}{Problem}{Problem}
\crefname{pluralproblem}{Problems}{Problems}

\crefname{assumption}{Assumption}{Assumption}
\crefname{pluralassumption}{Assumptions}{Assumptions}

\crefname{remark}{Remark}{Remark}
\crefname{pluralremark}{Remarks}{Remarks}

\crefname{proposition}{Proposition}{Proposition}
\crefname{pluralproposition}{Propositions}{Propositions}

\crefname{corollary}{Corollary}{Corollary}
\crefname{pluralcorollary}{Corollaries}{Corollaries}

\crefname{appendix}{Appendix}{Appendices}
\crefname{pluralappendix}{Appendices}{Appendices}

\newcommand{\Prob}{\amsmathbb{P}}
\newcommand{\R}{\amsmathbb{R}}

\newcommand{\N}{\amsmathbb{N}}

\newcommand{\distr}[1]{\ensuremath{\mathcal{P}(#1)}}
\newcommand{\distrd}[1]{\ensuremath{\mathcal{P}(#1)}}

\newcommand{\interior}[1]{\ensuremath{\mathrm{int(#1)}}}
\newcommand{\conv}{\ensuremath{\mathrm{conv}}}

\newcommand{\Prr}{\mathrm{Pr}}

\newcommand{\tuple}[1]{\ensuremath{( #1 )}}

\newcommand{\BackReach}{\ensuremath{\mathcal{R}^{-1}}}
\newcommand{\BackReachSt}{\ensuremath{\BackReach_{\mathrm{cl}}}}
\newcommand{\target}{\ensuremath{{d}}}
\newcommand{\partitionSpace}{\ensuremath{\mathcal{X}}}
\newcommand{\goalRegion}{\ensuremath{X_G}}
\newcommand{\unsafeRegion}{\ensuremath{X_U}}
\newcommand{\finiteHorizon}{\ensuremath{H}}
\newcommand{\region}{\ensuremath{\mathcal{V}}}
\newcommand{\inv}{\ensuremath{{-1}}}

\newcommand{\system}{\mathcal{S}}
\newcommand{\state}{{x}}
\newcommand{\control}{{u}}
\newcommand{\controlSpace}{U}
\newcommand{\controller}{\ensuremath{c}} 
\newcommand{\noise}{{\eta}}

\newcommand{\States}{S}

\newcommand{\Actions}{\ensuremath{\mathcal{A}}}

\newcommand{\transfunc}{P}
\newcommand{\transfuncImdp}{\ensuremath{\mathcal{P}}}
\newcommand{\Plb}{\ensuremath{\check{P}}}
\newcommand{\Pub}{\ensuremath{\hat{P}}}

\newcommand{\imdp}{\ensuremath{\mathcal{I}}}

\newcommand{\policy}{\ensuremath{\pi}}
\newcommand{\policySpace}{\ensuremath{\Pi}}

\newcommand{\Gauss}[2]{\ensuremath{\mathcal{N}(#1 , #2)}}
\newcommand{\mean}{\ensuremath{\mu}}
\newcommand{\cov}{\ensuremath{\Sigma}}

\newcommand{\pproperty}{\ensuremath{\varphi}}

\DeclareMathOperator*{\argmax}{arg\,max}

%% file: ancillary/colors.tex
\definecolor{labelA}{rgb}{0.992156863,0.682352941,0.380392157}
\definecolor{empty_color}{rgb}{1,1,0.749019608}
\definecolor{nonempty_color}{rgb}{0.670588235,0.866666667,0.643137255}

\tikzset{empty label/.style={draw, fill=empty_color}}
\tikzset{nonempty label/.style={draw, fill=nonempty_color}}
\tikzset{no label/.style={draw}}

\tikzset{small/.style={draw, inner sep=1pt, minimum size=1pt, font=\scriptsize}}
\tikzset{medium/.style={draw, inner sep=1pt, minimum size=15pt, font=\scriptsize}}

%% file: sections/1-Introduction.tex
\section{Introduction}
\label{sec:introduction}

The automated synthesis of correct-by-construction controllers for stochastic dynamical systems is crucial for their deployment in safety-critical scenarios. 
Synthesizing such controllers is challenging due to continuous and stochastic dynamics and the complexity of control tasks~\cite{DBLP:journals/automatica/LavaeiSAZ22}.
One solution is to abstract the system into a finite-state (also called symbolic) model~\cite{DBLP:journals/automatica/AbatePLS08,DBLP:journals/tac/LahijanianAB15,DBLP:books/daglib/0032856}.
Under an appropriate behavioral relation (e.g., a feedback refinement relation~\cite{DBLP:journals/tac/ReissigWR17}), trajectories of the abstraction are related to those of the dynamical system.
Thus, a controller (i.e., policy) in the abstraction can, by construction, be refined to a controller for the original system.

Conventional abstraction methods, however, rely on a precise mathematical model of the system.
In relaxing this assumption, \emph{data-driven abstractions} have recently gained momentum~\cite{DBLP:conf/adhs/MakdesiGF21,DBLP:journals/csysl/CoppolaPM23,DBLP:journals/csysl/LavaeiF23,DBLP:journals/corr/abs-2206-08069,DBLP:journals/automatica/HashimotoSKUD22,DBLP:conf/hybrid/SadraddiniB18, DBLP:conf/cdc/DevonportSA21,DBLP:journals/corr/abs-2303-17618}. These methods take a black-box (or sometimes a gray-box, e.g.,~\cite{DBLP:journals/automatica/HashimotoSKUD22}) perspective and construct abstractions from sampled system trajectories.
Several papers provide \emph{probably approximately correct} (PAC) guarantees~\cite{DBLP:journals/csysl/CoppolaPM23,DBLP:conf/cdc/DevonportSA21}, whereas others return controllers with hard (non-statistical) guarantees~\cite{DBLP:journals/csysl/LavaeiF23,DBLP:conf/hybrid/SadraddiniB18,DBLP:conf/adhs/MakdesiGF21}.
However, except from~\cite{DBLP:journals/corr/abs-2303-17618}, none of these methods can handle stochastic systems.

In this paper, we take a middle route between these model-based and data-driven abstraction methods.
Specifically, we consider control problems for linear systems with \emph{known} deterministic dynamics but stochastic noise of an \emph{unknown distribution}.
This \emph{hybrid} setting is similar to~\cite{DBLP:conf/aaai/BadingsA00PS22,DBLP:journals/jair/BadingsRAPPSJ23}, which develop a method to construct abstractions with PAC guarantees by sampling of the noise.
However, due to their exhaustive discretization of the state space, the application to large-scale, industrial and realistic systems remains elusive. 

A promising way to improve scalability is to exploit classical system properties, such as stability~\cite{franklin2019feedback}.
For example,~\cite{DBLP:conf/hybrid/Girard07} shows that for any stable discrete-time linear system with input constraints, there exists an approximately bismilar finite abstraction of any desired precision.
Similar results hold for incrementally stable continuous-time switched~\cite{DBLP:journals/tac/GirardPT10} and nonlinear systems~\cite{DBLP:journals/tac/Tabuada08,DBLP:journals/automatica/PolaGT08}.
A related notion is that of incremental forward completeness, which enables the abstraction of nonlinear discrete-time systems~\cite{DBLP:journals/tac/ZamaniPMT12}.
We observe that these results guarantee \emph{the existence} of a certain type of abstraction \emph{if} the system is stable.
However, we postulate that stability may also be beneficial to construct finite-state abstractions with \emph{smaller} underlying graphs.

\input{figures/approach}

Thus, the question central to this paper is: ``How can the stability of a stochastic dynamical system be exploited to synthesize controllers via finite-state abstractions with smaller underlying graphs?''
Our approach builds upon the hybrid abstraction technique for discrete-time stochastic linear systems developed in \cite{DBLP:conf/aaai/BadingsA00PS22}. 
We consider tasks as \emph{reach-avoid properties}, i.e., reach a set of goal states while always avoiding unsafe states.
The control objective is to design a feedback controller such that the closed-loop system satisfies the reach-avoid task with at least a desired threshold probability.

\change{Inspired by~\cite{DBLP:conf/aaai/BadingsA00PS22}, we create an abstraction for discrete-time stochastic linear systems into an interval Markov decision process (iMDP) with PAC intervals of transition probabilities, which we compute using data-driven techniques for scenario programs with discarded constraints~\cite{DBLP:journals/arc/CampiCG21,DBLP:journals/tac/RomaoPM23}.
A defining characteristic of this abstraction is that each abstract action is associated with a \emph{fixed distribution} over (continuous) successor states.
By contrast, other abstractions typically associate each abstract action with a \emph{fixed control input}, such that the distribution over successor states depends on the precise continuous state where the abstract action is chosen.
With our approach, we avoid this issue at the cost of more restrictive assumptions on the dynamics (\cref{assumption:nonsingular}).}

Instead of abstracting the open-loop dynamics directly (as in \cref{fig:1layer}), we propose the two-layer feedback control design framework in \cref{fig:2layer}.
In this framework, we first stabilize the system with a linear feedback gain and then abstract the stabilized dynamics.
This approach delegates part of the control effort to the stabilizing controller, which allows us to further constrain the control input synthesized in the abstraction.
Especially if the stabilizing controller contributes to satisfying the reach-avoid task, we can reduce the number of edges in the graph of the abstraction significantly (by up to 90\%; see \cref{sec:experiments}) with negligible performance loss.

\subsubsection*{Contributions}
\change{As our main contribution, we extend the abstraction method from~\cite{DBLP:conf/aaai/BadingsA00PS22} to the two-layer abstraction in \cref{fig:2layer} and show that this new scheme can be used to construct abstractions with smaller underlying graphs.
We show that the formal relation induced by the abstraction from~\cite{DBLP:conf/aaai/BadingsA00PS22} carries over to our setting.
Our experiments exemplify the conditions necessary for a good trade-off between abstraction size and controller performance.
}

%% file: figures/approach.tex
\tikzstyle{block} = [draw, rectangle, minimum height=2em, minimum width=8em, text centered]
\tikzstyle{sum} = [draw, circle, node distance=1.5cm]
\tikzstyle{input} = [coordinate]
\tikzstyle{output} = [coordinate]
\tikzstyle{pinstyle} = [pin edge={to-,thin,black}]

\newcommand\xshift{0.75em}

\begin{figure}[t] 

\begin{subfigure}[b]{0.47\linewidth}
    \centering
    
    \begin{tikzpicture}[
        auto,>=latex',font=\scriptsize,
        colA/.style={fill=Pastel2-A},
        colB/.style={fill=Pastel2-B},
        colC/.style={fill=Pastel2-C},
        ]
    \node (system) [block, colA, align=center] {\textbf{Dynamical system} \\ $x^+ \coloneqq Ax+Bu+\noise$};
    \node (controller) [block, colB, align=center, below=0.5em of system] {\textbf{Symbolic controller} \\ $u \coloneqq \controller(x, k)$};
    \node (abstraction) [block, colC, align=center, below=1.2em of controller, minimum height=1.1em, minimum width=5em] {\textbf{Abstraction}};

    \node [input, name=inputA, left=4em of system, yshift=0.8em] {};
    \node [input, name=inputB, left=1.7em of system] {};
    \node [input, name=inputC, left=1em of abstraction] {};
    \node [input, name=inputAbove, above=1.2em of system, xshift=-3*\xshift] {};
    \node [output, name=out, right=1em of system] {};

    \draw [->] (inputAbove) -- node [pos=0.2, right, align=center] {Stochastic noise $\noise$} (system.north-|inputAbove);
    \draw [->] (inputB) -- node [pos=0.3, above] {$u \in \controlSpace$} (system.west);
    \draw [->] (inputC) -- node [left, align=right, xshift=-0.4em] {Noise \\ samples} (abstraction.west);

    \draw [-] (system) -- node [pos=0.7, above] {$x$} (out);
    \draw [->] (out) |- node [] {} (controller.east);
    \draw [-] (controller.west) -| node [] {} (inputB);

    \draw [->] ([xshift=-\xshift] controller.south) -- node [left, midway] {$x^+, k$} ([xshift=-\xshift] abstraction.north);
    \draw [->] ([xshift=\xshift] abstraction.north) -- node [right, midway] {Abstract action} ([xshift=\xshift] controller.south);
    
    \end{tikzpicture}
    
    \caption{Single-layer abstraction.}
    \label{fig:1layer}
\end{subfigure}
\begin{subfigure}[b]{0.49\linewidth}
    \centering

    \begin{tikzpicture}[
        auto,>=latex',font=\scriptsize,
        colA/.style={fill=Pastel2-A},
        colB/.style={fill=Pastel2-B},
        colC/.style={fill=Pastel2-C},
        colD/.style={fill=Pastel2-F},
    ]
    \node (system) [block, colA, align=center] {\textbf{Dynamical system} \\ $x^+ = Ax+Bu+\noise$};
    \node (stabilizer) [block, colD, align=center, below=0.5em of system] {\textbf{Stabilizing controller} \\ $-Kx$};
    \node (controller) [block, colB, align=center, below=0.5em of stabilizer] {\textbf{Symbolic controller} \\ $u' \coloneqq \controller(x, k)$};
    \node (abstraction) [block, colC, align=center, below=1.2em of controller, minimum height=1.1em, minimum width=5em] {\textbf{Abstraction}};

    \node [input, name=inputA, left=4em of system, yshift=0.8em] {};
    \node [input, name=inputB, left=1.7em of system] {};
    \node [input, name=inputC, left=1em of abstraction] {};
    \node [input, name=inputAbove, above=1.2em of system, xshift=-3*\xshift] {};
    \node [output, name=out, right=1em of system] {};

    \node [sum, name=sum, left=2em of stabilizer] {};

    \draw [->] (inputAbove) -- node [pos=0.2, right, align=center] {Stochastic noise $\noise$} (system.north-|inputAbove);
    \draw [->] (sum.north) -- (inputB-|sum.north) -- node [pos=0.3, above, align=center] {$u = -Kx $ \\ $ + u' \in \controlSpace$} (system.west);
    \draw [->] (inputC) -- node [left, align=right, xshift=-0.4em] {Noise \\ samples} (abstraction.west);
    
    \node [above=-0.4em of sum, xshift=0.5em] {$+$};

    \draw [-] (system) -- node [pos=0.7, above] {$x$} (out);
    \draw [->] (out) |- node [] {} (controller.east);
    \draw [->] (out|-stabilizer.east) |- node [] {} (stabilizer.east);
    \draw [->] (controller.west) -| node [pos=0.3, below, align=center] {$u' \in \controlSpace'$} (sum);
    \draw [->] (stabilizer.west) -- node [pos=0.3, above] {} (sum);

    \draw [->] ([xshift=-\xshift] controller.south) -- node [left, midway] {$x^+, k$} ([xshift=-\xshift] abstraction.north);
    \draw [->] ([xshift=\xshift] abstraction.north) -- node [right, midway] {Abstract action} ([xshift=\xshift] controller.south);
    
    \end{tikzpicture}
    
    \caption{Two-layer abstraction.}
    \label{fig:2layer}
\end{subfigure}
    \caption{A single-layer abstraction (a), versus our two-layer feedback design framework, which combines a stabilizing controller (linear feedback gain) with a symbolic controller (obtained from the abstraction).
    We impose constraints on both the total input $-Kx+u' \in \controlSpace$ and on the input $u' \in \controlSpace'$.}  
    \label{fig:approach}
\end{figure}
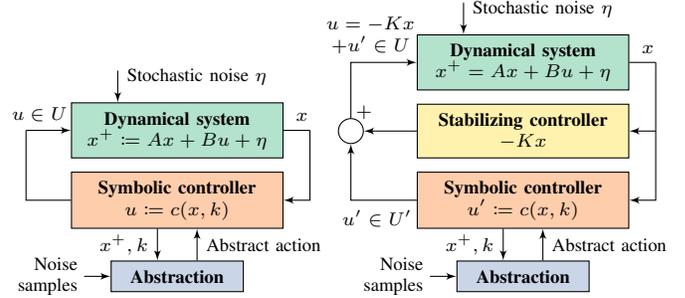

%% file: sections/2-Background.tex
\section{Preliminaries}
\label{sec:background}

\change{A probability space $\tuple{\Omega, \mathcal{F}, \Prob}$ consists of an uncertainty space $\Omega$, a $\sigma$-algebra $\mathcal{F}$, and a probability measure $\Prob \colon \mathcal{F} \to [0,1]$.
A random variable $x$ is a measurable function $x \colon \Omega \to \R^n$ for some $n \in \N$, which takes value $x(\omega) \in \R^n$ for $\omega \in \Omega$.}
We denote the set of all distributions for both a continuous and discrete set $X$ by $\distr{X}$.
The convex hull of a set of points $\{v_1, \ldots, v_m\}$ in $\R^n$ is $\conv(v_1,\ldots,v_n)$. 
We denote the interior of $V \subset \R^n$ by $\interior{V}$ and the pseudoinverse of matrix $B$ by $B^\dag$.
The indicator function $\mathbb{1}_V(x)$ for a set $V \subset \R^n$ is one if $x \in V$ and zero otherwise.
\change{The Cartesian product of an interval is written as $[a,b]^n$, for $a \leq b$, $n \in \N$.}

\subsection{Stochastic dynamical systems}

Consider a discrete-time, stochastic linear dynamical system $\system$ where the state space variable $x_k \in \R^n$ evolves as
\begin{equation}
    \label{eq:linear_system}
    \system : \,\, x_{k+1} = A\state_k + B\control_k + \noise_k, \quad x_0 = \bar{x},
\end{equation}
where $\bar{x} \in \R^n$ is the initial condition, $\control_k \in \R^p$ is the control input, and $\noise_k \in \R^n$ is a stochastic noise. Matrices $A$ and $B$ have the appropriate dimensions.
The control input is constrained to a convex polytope $\controlSpace = \{ u \in \R^p : Gu \leq g\} \subset \R^p$ called the admissible control input, where $G \in \R^{q \times p}$ and $g \in \R^q$. 
Moreover, $(\noise_k)_{k \in \N_0}$ is a discrete-time stochastic process defined on a probability space $(\Omega,\mathcal{F},\Prob)$, with its natural filtration (see \cite{Durrett96}~for details). 
Thus, $(\state_k)_{k \in \N_0}$ is also a stochastic process in the same probability space.
%
\begin{assumption}[Non-singular and controllable]
    \label{assumption:nonsingular}
    Matrix $A \in \R^{n \times n}$ is non-singular, and the pair $(A,B) $ is controllable.
\end{assumption}
\begin{assumption}[Noise distribution]
    \label{assumption:unknown_noise}
    For all $(x,u) \in \R^n \times \controlSpace$ and all (Borel measurable~\cite{Salamon16}) sets $V \subset \R^n$, let $\mu_k(V;x,u) = \Prob\{ \omega \in \Omega : Ax + Bu + \noise_k(\omega) \in V\} \in [0,1]$ be the conditional probability that the next state belongs to $V$, given the current state-input pair. Then we have that:
    \begin{itemize}
        \item (\textit{identically distributed}): $\mu_k(V;x,u)$ is time-invariant; hence, we may drop the time index and write $\mu(V;x,v)$;
        \item (\textit{independence}): For any finite collection $V_1,\ldots,V_m \subset \R^n$ and state-action pairs $\{(x_i,u_i)\}_{i =1}^m$, we have that 
        {\small
        \begin{equation*}
            \Prob\{\omega \in \Omega : \bigcap_{i=1}^m Ax_i + Bu_i + \noise(\omega) \in V_i \} = \prod_{i = 1}^m \mu(V_i; x_i,u_i);
        \end{equation*}
        }
        \item (\textit{density}): The Radon-Nikodym derivative of $\mu(V;x,u)$ exists for all pairs $(x,u)$, and $\mu(V';x,u) $ is a measurable function from $\R^n\times U$ to $[0,1]$ for all $V' \subset \R^n$.
        However (and importantly), the density $\mu$ is unknown.
    \end{itemize}
\end{assumption}
\change{\cref{assumption:unknown_noise} requires process $(\noise_k)_{k \in \N}$ to be i.i.d. and to possess a well-defined probability density. 
However, we do not assume knowledge of its density.}
Under  \cref{assumption:unknown_noise}, system \eqref{eq:linear_system} can be equivalently expressed using the stochastic kernel (see \cite[Chapter~7]{Bertsekas.Shreve78}~for details) $T \colon \R^n \times \controlSpace \to \distr{\R^n}$,
which maps each state-input pair to a distribution over states:
\begin{equation}
    \label{eq:linear_system_kernel}
    x_{k+1} \sim T(\cdot \mid x_k,u_k), \quad x_0 = \bar{x}.
\end{equation}
For example, if $\noise_k \sim \Gauss{\mean}{\cov}, \, k \in \N$ is Gaussian, the stochastic kernel is given by $T(\cdot \mid x_k,u_k) = \Gauss{Ax_k + Bu_k + \mean}{\cov}$.
A time-varying feedback controller chooses the inputs $\control_k \in U$ by measuring the current state.
%
\begin{definition}
    \label{def:controller}
    A time-varying feedback controller is a measurable function $\controller \colon \R^n \times \N_0 \to \controlSpace$, which maps a state $x \in \R^n$ and a time step $k \in \N_0$ to a control input $u \in \controlSpace$.
\end{definition}

\subsection{Problem statement}

Given a system as in \eqref{eq:linear_system}, our goal is to find a time-varying controller $c$ such that the closed-loop system with $u_k = c(x_k,k)$, for all $k \in \N$, satisfies some objective.
Specifically, we consider the objective of reaching a desired region $\goalRegion$ of the state space while always avoiding unsafe states $\unsafeRegion$.
\begin{definition}[Reach-avoid property]
    \label{def:ReachAvoidProperty}
    A reach-avoid property is a tuple $\tuple{\goalRegion, \unsafeRegion, \finiteHorizon}$ of a set of goal states $\goalRegion$ and unsafe states $\unsafeRegion$ (such that $\goalRegion \cap \unsafeRegion = \varnothing$), and horizon $\finiteHorizon \in \N$.
\end{definition}

For system $\system$ in \eqref{eq:linear_system}, we consider $\goalRegion, \unsafeRegion \subset \R^n$ as compact subsets of $\R^n$ and use the notation $\pproperty = (\goalRegion,\unsafeRegion, \finiteHorizon)$ for this reach-avoid property.
A trajectory $\state_0, \state_1, \ldots, x_\finiteHorizon$ of length $\finiteHorizon$ for system $\system$ satisfies $\pproperty$ if there exists a $k \in \{0,\ldots,\finiteHorizon\}$ such that $x_k \in \goalRegion$ and $x_{k'} \notin \unsafeRegion$ for all $k' \in \{0,\ldots,k\}$.
Under a fixed controller, system $\system$ induces a stochastic process $(x_k)_{k \in \N_0}$ for which we can reason over the probability of satisfying a reach-avoid property~\cite{Bertsekas.Shreve78}. 
\begin{definition}[Satisfaction of $\pproperty$]
\label{def:satisfaction}
For a fixed time-varying feedback controller $c:\R^n\times \N \rightarrow \controlSpace$ and a given initial condition $\bar{\state}$, the satisfaction probability of $\pproperty$ is denoted by
\begin{align}%
\Prr_\system^c (\bar{\state} \models \pproperty) \coloneqq
\Prob \Big\{ &
     \omega \in \Omega
    :
    \exists k \in \{0,\ldots,\finiteHorizon\}, \state_k(\omega) \in \goalRegion, 
    \nonumber
    \\ & \state_{k'}(\omega) \notin \unsafeRegion \,\, \forall k' \in \{0,\ldots,k \}, 
    \\ & \state_0 = \bar{\state}, \,
    x_{k + 1} \sim T(\cdot\mid x_k, c(x_k,k))
    \nonumber
\Big\}.
\end{align}
\end{definition}

We are now able to describe our control objective.
\begin{problem}
\label{prob:Problem}
Given a linear stochastic system $\system$ as in \eqref{eq:linear_system}, with initial state $\bar{\state}$, a reach-avoid property $\pproperty$ as in Definition \ref{def:ReachAvoidProperty}, and a desired threshold probability $\rho \in [0,1]$, design a time-varying feedback controller $\controller$ such that $\Prr_\system^c(\bar{\state}\models \pproperty) \geq \rho$.
\end{problem}

\subsection{Markov decision processes}

Our approach for solving \cref{prob:Problem} is to create a discrete abstraction of system $\system$ into an MDP with imprecise transition probabilities.
To distinguish from system $\system$, we call abstract states \emph{locations} and a controller for the abstraction a \emph{policy}.
\begin{definition}
\label{def:iMDP}
\change{An \emph{interval MDP} (iMDP) $\imdp$ is defined as a tuple $\imdp \coloneqq \tuple{ S, \bar{s}, \Actions, \Plb, \Pub }$, where}
\begin{itemize}
    \item $S$ is a finite set of locations, with initial condition $\bar{s} \in S$,
    \item $\Actions$ is a finite set of actions, with $\Actions(s) \subset \Actions$ denoting the actions enabled in location $s \in S$, and
    \item $P \colon S \times \Actions \rightrightarrows \distrd{S}$ maps each pair $(s,a)$ to a set of distributions defined by $\Plb(s,a), \Pub(s,a) \in [0,1]^{|\States|}$ as 
        \begin{align}
            P(s,a) = \Big\{ p \in [0,1]^{|S|} : \, &\Plb_{s'}(s,a)   \leq p_{s'} \leq \Pub_{s'}(s,a),   \nonumber 
            \\  & \forall s' \in S,~\sum_{s' \in S} p_{s'} = 1 \Big\}.
            \label{eq:iMDP-transition}
        \end{align}
\end{itemize}
\end{definition}

For any iMDP, we require that $\Plb_{s'}(s,a) \leq \Pub_{s'}(s,a)$ for all $s, s' \in S, a \in \Actions(s)$, and that $\sum_{s' \in  S} \Plb_{s'}(s,a) \leq 1 \leq \sum_{s' \in S} \Pub_{s'}(s,a)$ for all $s \in S, a \in \Actions(s)$; otherwise, the set \eqref{eq:iMDP-transition} may be empty. 
\change{An \emph{adversary} fixes a probability $P'(s,a)(s') \in P(s,a)$ for all pairs $(s,a) \in \States \times \Actions$.
Importantly, a different $P'$ can be chosen every time the same pair $(s,a)$ is encountered.
For brevity, we overload notation and use $P' \in P$ to denote choosing an adversary in the set of all adversaries.
}


Actions are chosen by a time-varying policy $\policy \colon S \times \N_0 \to \distrd{\Actions}$, which maps every location $s \in S$ and time $k \in \N_0$ to an action $a \in \Actions$.\footnote{Time-varying policies are needed to attain optimal values for the time-bounded properties we consider~\cite[Ch.~10.6]{DBLP:books/daglib/0020348}. An equivalent approach is to encode the time step explicitly in the iMDP by defining the set of locations $S' = S \times \{0,\ldots,\finiteHorizon\}$ and using memoryless policies $\policy \colon S' \to \Actions$ instead.}
The set of all admissible policies\footnote{The policy class $\policySpace$ suffices to obtain optimal policies for iMDP~\cite{DBLP:conf/cav/PuggelliLSS13}.} is 
{\small
$$
    \policySpace = \big\{ \policy \colon S \times \N_0 \to \distrd{\Actions} \, \mid \, \policy(s,k)(a) > 0 \implies a \in \Actions(s) \big\}.
$$}%
Thus, any policy $\policy \in \policySpace$ requires that for all $k \in \N$ and $s\in S$, the support of the distribution $\policy(s,k)$ is contained in $\Actions(s)$.

For an iMDP, a reach-avoid property $\pproperty' = (S_G, S_U, \finiteHorizon)$ (cf.~\cref{def:ReachAvoidProperty}) is defined over the locations, i.e., $S_G, S_U \subseteq S$.
The semantics over trajectories $s_0, s_1,\ldots, s_{\finiteHorizon}$ are the same as for system $\system$. 
Similar to \cref{def:satisfaction}, for any policy $\policy \in \policySpace$ and transition function $P' \in P$, we denote the probability of satisfying $\pproperty'$ by $\Prr_{P'}^\policy(\bar{s} \models \pproperty')$.\footnote{Fixing a policy $\policy \in \policySpace$ and an adversary with $P' \in P$ induces a Markov chain with probability measure $\Prr_{P'}^\policy$; see~\cite[Def.~10.10]{DBLP:books/daglib/0020348} for details.}
An optimal (robust) policy $\policy^\star \in \policySpace$ optimizes the next min-max problem:
\begin{equation}
    \label{eq:optimal_policy}
    \policy^* \in \argmax_{\policy \in \policySpace} \, \min_{P' \in P} \Prr_{P'}^\policy(
    \bar{s} \models \pproperty').
\end{equation}

\begin{remark}
\label{remark:rewards}
We can alternatively express reach-avoid properties by extending the iMDP with a reward function $R \colon S \to \R_{\geq 0}$ defined as $R(s) = \mathbb{1}_{S_G}(s)$, and making all locations $s \in S_U$ absorbing, i.e., $\Plb(s,a,s) = \Pub(s,a,s) = 1 \,\forall s \in S_U$, $a \in \Actions(s)$.
For details, we refer to~\cite[Def.~10.71]{DBLP:books/daglib/0020348}.
\end{remark} 

%% file: sections/3-Abstraction.tex
\section{Abstraction-Based Controller Synthesis}
\label{sec:abstraction}

We formally relate the dynamics in \eqref{eq:linear_system} to a finite iMDP, using a probabilistic variant of a feedback refinement relation~\cite{DBLP:journals/tac/ReissigWR17}.
Then, we establish that the abstraction proposed by~\cite{DBLP:conf/aaai/BadingsA00PS22} induces this relation.
A measurable set $R \subseteq \R^n \times S$ is called a binary relation, for which we use notations $R(x) = \{ s \in S: (x,s) \in R \}$ and $R^\inv(s) = \{ x \in \R^n : (x,s) \in R \}$.

\begin{definition}[{\cite{DBLP:journals/jair/BadingsRAPPSJ23}}]
    \label{def:relation}
    A binary relation $R\subset \R^n \times S$ is a \emph{probabilistic feedback refinement relation} from iMDP $\imdp = \tuple{S,\bar{s},\Actions,\Plb,\Pub}$ to system $\system$ defined by \eqref{eq:linear_system_kernel} if
    \begin{enumerate}
        \item for the initial state-location, we have $(\bar{\state}, \bar{s}) \in R$, and
        \item for all $(\state, s) \in R$ and $a \in \Actions(s)$, there exists a $u \in \controlSpace$ such that for all $s' \in S$, it holds that
    \end{enumerate}
    \begin{align}
        \label{eq:relation_bounds}
        &\Plb_{s'}(s,a) 
        \leq \int_{\R^n} \mathbb{1}_{R^\inv(s')}(\xi) T( d\xi \mid x,u) 
        \\
        &\enskip = \Prob\big\{ \omega \in \Omega: Ax + Bu + \noise(\omega) \in R^{-1}(s') \big\}
        \leq \Pub_{s'}(s,a).
        \nonumber
    \end{align}
\end{definition}

Similar to~\cite{DBLP:journals/tac/ReissigWR17}, we denote a probabilistic feedback refinement relation $R$ from $\imdp$ to $\system$ by $\imdp \preceq_R \system$.
Moreover, we also use the relation $R$ in \cref{def:relation} to relate reach-avoid properties between system $\system$ and an iMDP.
\begin{definition}
    \label{def:property_relation}
    A pair of reach-avoid properties $\pproperty = \tuple{\goalRegion, \unsafeRegion, \finiteHorizon}$ and $\pproperty'= \tuple{S_G, S_U, \finiteHorizon}$ is consistent under a relation $R \subset \R^n \times S$, denoted by $\pproperty' \preceq_R \pproperty$ if
    \begin{align}
        S_G &= \{s \in S : R^\inv(s) \subseteq \goalRegion \},
        \\
        S_U &= \{s \in S : R^\inv(s) \cap \unsafeRegion \neq \varnothing \}.
    \end{align}
\end{definition}

Intuitively, given an iMDP path $s_0, s_1, \ldots, s_\finiteHorizon$ that satisfies $\pproperty'$, the relation $\pproperty' \preceq_R \pproperty$ implies that \emph{all related trajectories} $x_0, x_1, \ldots, x_\finiteHorizon$, i.e., trajectories for which $(x_i, s_i) \in R \,$ for all $ i = 0,\ldots,\finiteHorizon$, must satisfy $\pproperty$.
The following result, which is proven in~\cite{DBLP:journals/jair/BadingsRAPPSJ23}, shows that \cref{def:relation,def:property_relation} can be used to synthesize correct-by-construction controllers for system $\system$.

\begin{theorem}[{\cite{DBLP:journals/jair/BadingsRAPPSJ23}}]
    \label{thm:existence_controller}
    \change{Consider a system $\system$ as in \cref{eq:linear_system}, an iMDP as in \cref{def:iMDP}, and a relation $R \subset \R^n \times S$ such that $\imdp \preceq_R \system$.
    Also let properties $\pproperty = \tuple{\goalRegion, \unsafeRegion, \finiteHorizon}$ and $\pproperty' = \tuple{S_G, S_U, \finiteHorizon}$ be such that $\pproperty' \preceq_R \pproperty$.
    Then, for any policy $\pi \in \policySpace$, there exists a controller $c$ as in \cref{def:controller} such that}
    \begin{equation}
        \label{eq:existence_controller}
        \Prr_\system^c(\bar{x} \models \pproperty) \geq
        \min_{P' \in P} \Prr_{P'}^\policy(\bar{s} \models \pproperty').
    \end{equation}
\end{theorem} 

The proof of \cref{thm:existence_controller} uses that, for MDPs, the relation $R$ preserves the satisfaction of probabilistic computation tree logic (PCTL), in which the reach-avoid property in \cref{def:ReachAvoidProperty} can be expressed~\cite{DBLP:journals/iandc/HermannsPSWZ11}.
For iMDPs, this preservation of probabilistic satisfaction leads to the inequality in \eqref{eq:existence_controller}.
In fact, \cref{thm:existence_controller} can be extended to any PCTL formula; see, e.g.,~\cite{DBLP:conf/qest/RickardBRA23}.
However, since our contributions are unrelated to the property, we focus on reach-avoid properties for simplicity.

\subsection{Abstraction procedure}

We revisit the abstraction developed in~\cite{DBLP:conf/aaai/BadingsA00PS22,DBLP:journals/jair/BadingsRAPPSJ23}, which first uses a model-based approach to compute the abstract locations and actions.
Second, a data-driven approach is used to capture the stochastic uncertainty into intervals of transition probabilities.
The resulting abstraction is an iMDP that creates a relation as in \cref{def:property_relation} with a pre-defined confidence level.

\subsubsection{Model-based locations and actions}
The locations of the abstraction are given by a polyhedral partition of a bounded portion $\partitionSpace \subset \R^n$ of the state space of system $\system$:
\begin{definition}
    A polyhedral partition of $\partitionSpace \subset \R^n$ is a finite collection of sets $\{\region_1,\ldots, \, \region_L, \R^n \setminus \partitionSpace\}$ such that 
    \begin{enumerate}
        \item Each $\region_i$ is a convex polytope, i.e., $\region_i = \{x \in \R^n : H_i x \leq h_i \}$ for $H_i \in \R^{p_i \times n}$, $h_i \in \R^{p_i}$, and $p_i \in \N$;
        \item $\partitionSpace = \bigcup_{i=1}^L \region_i$;
        \item $\interior{\region_i} \bigcap \interior{\region_j} = \emptyset, \,\, \forall i,j \in \{1,\ldots,L\}, \,\, i \neq j$.
    \end{enumerate}
\end{definition}

Adding the final element $\R^n \setminus \partitionSpace$ ensures that the partition covers $\R^n$.
A partition creates an equivalence relation~\cite{belta2017formal}.
\begin{remark}[Equivalence relation]
    \label{remark:induced_relation}
    A polyhedral partition of $\partitionSpace \subset \R^n$ creates an equivalence relation ${\sim} \subset \R^n \times \R^n$, such that $[x]_\sim \coloneqq \{ x' \in \R^n \mid x \sim x' \}$ is the equivalence class of state $x \in \R^n$, where $x \sim x'$ denotes that $(x,x') \in {\sim}$.
    The set of all equivalence  classes $\R^n /{\mathord{\sim}} = \{ [x]_\sim \mid x \in \R^n \} = \{\region_1,\ldots,\region_L, \, \R^n \setminus \partitionSpace\}$ is the partition itself.
\end{remark}
The locations of the abstraction are the equivalence classes of the partition, i.e., $S \coloneqq \R^n /{\mathord{\sim}}$.
Next, the set of actions is $\Actions \coloneqq \{a_1, \ldots, a_q\}$, $q \in \N$, where each $a \in \Actions$ is associated with a target point $\target_a \in \R^n$ in the state space of $\system$.
For each point $\target_a$, for $a \in \Actions$, we define the backward reachable set as
\begin{equation}
\begin{split}
    \label{eq:backward_reachable_set}
    \BackReach(d_a) &= \{ x \in \R^n : d_a = Ax + Bu, \, u \in \controlSpace \}
    \\
    &= \mathrm{conv}\big( A^\inv(\target_a - B v^i) \,:\, i = 1,\ldots,q \big),
\end{split}
\end{equation}
where the second equality follows from \cref{assumption:nonsingular}.
The set $\Actions(s) \subseteq \Actions$ of actions enabled in location $s \in S$ is
\begin{equation}
    \label{eq:enabled_actions}
    \Actions(s) = \big\{ 
        a \in \Actions \mid
        s \subseteq 
        \BackReach(\target_a)
    \big\}.
\end{equation}
Thus, action $a \in \Actions$ is enabled in location $s \in S$ only if the equivalence class $s$ is contained in the backward reachable set $\BackReach(\target_a)$.
\change{Choosing abstract action $a \in \Actions$ is defined such that $\target_a = Ax + Bu$, which implies that $u = B^\dag(\target_a - Ax)$.\footnote{Action $a$ is only enabled in equivalence classes that are a subset of the backward reachable set $\BackReach(\target_a)$.
Thus, we have $u = B^\dag(d_a-Ax) \in \controlSpace$ by construction for any state $x \in \bigcup \big\{ s \in S \mid a \in \Actions(s) \big\} \subset \R^n$.}
Since the noise is additive, the successor state is $\target_a + \noise$, which is a random variable with distribution $T(\cdot \mid x, B^\dag(\target_a - Ax) )$.
}%
\begin{remark}
\change{
Other abstraction methods typically associate each $a \in \Actions$ with a fixed input $\hat{u} \in \controlSpace$.
Thus, the distribution $T(\cdot \mid x, \hat{u} )$ over successor states associated with choosing action $a$ depends on the precise state $x \in \R^n$.
By contrast, we associate each abstract action $a \in \Actions$ with a fixed distribution $T(\cdot \mid x, B^\dag(\target_a - Ax) )$ over successor states.
Since $Ax + Bu + \noise = Ax + B B^\dag(d_a-Ax) + \noise = \target_a + \noise$, this distribution is the same for any state $x \in \R^n$ for which $a \in \Actions([x])$, i.e., where abstract actions $a$ is enabled.
}
\end{remark}

\subsubsection{Data-driven transition probabilities}
As the distribution of the noise is unknown, we use a finite set of samples of $\noise_k$ to compute an interval on the probability of reaching each location $s \in S$.
Formally, let $\{\delta_1,\ldots,\delta_N\} \in \Omega^N$, where  $N \in \N$ is the number of samples\footnote{Since \eqref{eq:linear_system} is time-invariant and has additive noise, we can obtain these samples from a \textit{single} trajectory of length $N$ starting at an arbitrary state $\bar{x}$.} of $\noise_k$.
For each pair $s, s' \in S$ and enabled action $a \in \Actions(s)$, the interval $[\Plb_{s'}(s,a), \Pub_{s'}(s,a)]$ is computed such that the exact probability is contained with at least a desired probability of $1-\beta$, with $\beta \in (0,1)$:
\begin{align}
    \nonumber
    \Prob^N \Big\{ \,\,
    \Plb_{s'}(s,a) 
    &\leq 
    \int_{\R^n} \mathbb{1}_{s'}(\xi) T( d\xi \mid x, \, [B^\dag(d_a-Ax)] )
    \\
    &\leq
    \Pub_{s'}(s,a)
    \,\, \Big\} \geq 1-\beta,
    \label{eq:iMDP_intervals}
\end{align}
\change{where $x$ is such that $[x] = s$, state $s' \in \States = \R^n /{\mathord{\sim}}$ is interpreted as a subset of $\R^n$,} and the outer probability is taken w.r.t. the upper bound $\Pub_{s'}(s,a)$ and lower bound $\Plb_{s'}(s,a)$ of the interval, which are random variables in the space $\Omega^N$.

In practice, we compute these intervals using the method from~\cite{DBLP:journals/jair/BadingsRAPPSJ23}, which leverages the scenario approach~\cite{DBLP:journals/arc/CampiCG21}.
\change{This method implicitly solves a set of $2N$ convex scenario programs with discarded constraints~\cite{DBLP:journals/jota/CampiG11} and uses~\cite{DBLP:journals/tac/RomaoPM23} to compute tight bounds on the probability of \emph{constraint violation} for each of these programs.
By construction, one of these $2N$ probabilities of constraint violations lower bounds the transition probability in \cref{eq:iMDP_intervals}, and another one is an upper bound.
By choosing the confidence level on the probability of constraint violation for each scenario program as $1-\frac{\beta}{2N}$, we obtain a probability interval such that \cref{eq:iMDP_intervals} holds.}
It is shown in~\cite{DBLP:journals/jair/BadingsRAPPSJ23} that these scenario programs can be solved analytically based on its geometry, making the approach highly efficient. 
Due to space restrictions, we refer to~\cite[Theorem~1]{DBLP:journals/jair/BadingsRAPPSJ23} for formal details.

\subsubsection{Complete abstraction}
\label{sec:abstraction_full}
Putting all elements together, we define the iMDP abstraction $(S, \bar{s}, \Actions, \Plb, \Pub)$,
with
\begin{itemize}
    \item Set of locations $S = \R^n /{\mathord{\sim}}$, with $\bar{s} = [\bar{x}]$;
    \item Actions $\Actions = \{a_1,\ldots,a_q\}$, for $q \in \N$, with the enabled actions $\Actions(s)$ defined by \eqref{eq:enabled_actions} for all $s \in S$;
    \item For each $s,s' \in S$ and $a \in \Actions(s)$, the lower and upper bound probabilities $\Plb_{s'}(s,a)$ and $\Pub_{s'}(s,a)$ are such that \eqref{eq:iMDP_intervals} holds for a desired value of $\beta \in (0,1)$.
\end{itemize}

\subsection{Controller synthesis}
\change{We show that the equivalence relation ${\sim} \subset \R^n \times \States$ created by the partition (cf. \cref{remark:induced_relation}) is (with a certain probability) a probabilistic feedback refinement relation $R$ from iMDP $\imdp$ to system $\system$, as defined by the conditions in \cref{def:relation}.}
%
\begin{theorem}[{\cite[Thm.~2]{DBLP:journals/jair/BadingsRAPPSJ23}}]
    \label{thm:pac_relation}
    For a given polyhedral partition that creates an equivalence relation $\sim$, let $\imdp$ be the iMDP abstraction for system $\system$ with $\beta \in (0,1)$.
    Then, it holds that $\Prob^N\{ \imdp \preceq_\sim \system \} \geq 1-\beta \cdot |\Actions| \cdot |S|$.
\end{theorem}
\begin{proof}
    The iMDP has at most $|\Actions| \cdot |S|$ unique probability intervals (see \cite{DBLP:journals/jair/BadingsRAPPSJ23} for details). 
    We have that $\imdp \preceq_\sim \system$ if all of these intervals contain the exact probability, which (by applying the union bound) is satisfied with a probability of at least $1 - \beta \cdot |\Actions| \cdot |S|$.
    Thus, the claim follows.
\end{proof}

Under \cref{assumption:nonsingular}, we can refine any policy for the abstract iMDP into a controller of the form in \cref{def:controller}.
\begin{definition}[Controller refinement]
    \label{def:refined_controller}
    Let $\policy \in \policySpace$ be any iMDP policy. 
    The refined controller $\controller \colon \R^n \times \{0,\ldots,\finiteHorizon\} \to \controlSpace$ is piece-wise affine in $x \in \R^n$ and is defined for all $x \in \R^n$ as
    \begin{equation}
        \label{eq:refined_controller}
        \controller(\state, k) = B^\dag (\target_a - A\state), \enskip
        a \in \policy(s,k) \in \Actions(s),
    \end{equation}
    where $s \in S$ is the iMDP location such that $[x]_\sim \in s$.\footnote{If $x \in \R^n$ is on the boundary of multiple partition elements, the refined controller can select any location $s = [x]_\sim \in S$.}
\end{definition}

Finally, we obtain the following key result for the iMDP abstraction and the refined controller defined above.
\begin{theorem}
    \label{thm:pac_control_design}
    Let $\system$ be a stochastic linear system $\system$, $\pproperty$ a reach-avoid property, and $\sim$ the equivalence relation for a polyhedral partition.
    Then, for the iMDP abstraction $\imdp$, a reach-avoid property $\pproperty_\imdp$ such that $\pproperty_\imdp \preceq_\sim \pproperty$, and any policy $\pi \in \policySpace_\imdp$ with refined controller $c$ (as per \cref{def:refined_controller}), it holds that
    \begin{equation*}
        \Prob^N \Big\{
        \,
        \Prr_\system^{\controller}(x_0 \models \pproperty)
        \geq
        \min_{\transfunc \in \transfuncImdp} \Prr_{\imdp[P]}^\policy(\bar{s} \models \pproperty_\imdp)
        \,
        \Big\} \geq 1-\beta \cdot |\Actions| \cdot |S|.
    \end{equation*}
\end{theorem}

Thus, the satisfaction probability on the iMDP is a \emph{lower bound} on the satisfaction probability for system $\system$ under the refined controller, with probability at least $1-\beta \cdot |\Actions| \cdot |S|$.

%% file: sections/4-Stability.tex
\section{Exploiting Stability in Abstraction}
\label{sec:stability}

\change{The size of the iMDP abstraction (which can be expressed by the number of edges, or transitions, in the underlying graph) from \cref{sec:abstraction} grows exponentially with the dimension of the state space.}
In this section, we develop an extension to the method from~\cite{DBLP:conf/aaai/BadingsA00PS22} to create smaller abstractions.
As our key contribution, we leverage the two-layer control design framework in \cref{fig:2layer}, which first stabilizes the dynamics and then creates an abstraction of the closed-loop dynamics.
Specifically, we use the feedback control law given by
\begin{equation}
    \control_k = -K x_k + \control'_k,
    \label{eq:two-layer-control-law}
\end{equation}
where the gain matrix $K \in \R^{m \times n}$ represents a stabilizing control law, and $u'$ is the control input captured by the abstraction.
In this paper, we obtain the feedback gain matrix by solving an instance of a linear quadratic regulator (LQR) \cite{franklin2019feedback} control problem.
Applying the feedback control law in \eqref{eq:two-layer-control-law} to system \eqref{eq:linear_system} yields the closed-loop dynamics given by
\begin{equation}
    \label{eq:dynamics_double_input}
    \state_{k+1} = 
    A_\mathrm{cl} \state_k + B \control'_k + \noise_k,
\end{equation}
where $A_\mathrm{cl} = A-BK$. 
We assume that the feedback gain $K$ satisfies the input constraints in the following way.
\begin{assumption}    
    \label{assumption:admissibility}
    The gain matrix $K \in \R^{m \times n}$ is such that $-K x \in \controlSpace$ for all $x \in \partitionSpace$ and the matrix $A_{\mathrm{cl}}$ is non-singular.
\end{assumption}

\subsection{Backward reachable set for stabilized dynamics}
We show how the iMDP abstraction described in \cref{sec:abstraction} can be employed together with the two-layer feedback control law in \eqref{eq:two-layer-control-law}. 
The key step is that we modify the backward reachable set computation in \eqref{eq:backward_reachable_set}, replacing it by
\begin{align}
    \label{eq:backward_reachable_set_stable}
    \BackReachSt(\target_a, \controlSpace') = \big\{
    \state \in \R^n : \,\,
    &\target_a = A_\mathrm{cl} \state + B\control',
    \\
    &-K\state + \control' \in \controlSpace,
    \,
    \control' \in \controlSpace'
    \big\},
    \nonumber
\end{align}
where the constraint $-K \state + \control' \in \controlSpace = \{ u \in \R^m : Gu \leq h\}$ enforces that the total input $u$ is admissible, and the constraint $\control' \in \controlSpace' = \{ u \in \R^m : G'u \leq h'\}$ controls the size of the abstraction. Matrices $G$ and $G'$ and vectors $h$ and $h'$ define the admissible control inputs; their sizes are omitted for brevity.

\begin{assumption}
    \label{assumption:origin_in_barU}
    The set $\controlSpace'$ 
    contains the origin, i.e., $0 \in \controlSpace'$.
\end{assumption}

Observe that \cref{eq:backward_reachable_set_stable} is of the same form as \cref{eq:backward_reachable_set} (despite imposing additional constraints) and can thus be computed similarly, as shown by the following lemma.
\begin{lemma}
    \label{lemma:New-backward-set}
    Under Assumptions \ref{assumption:admissibility} and \ref{assumption:origin_in_barU}, the following holds:
    \begin{itemize}
    \item[$i)$] For any $\target_a \in \R^n$, the set $\BackReachSt(\target_a, U')$ is non-empty;
    \item[$ii)$] $\BackReachSt(\target_a, \controlSpace') = \{ x \in \R^n \mid \target_a = A_\mathrm{cl}x + Bu', \, u' \in \tilde\controlSpace \}$, 
    where $\tilde\controlSpace \subset \R^m$ is a convex polytope defined as
    \begin{equation}
        \label{eq:new_control_set}
        \tilde\controlSpace = \big\{ u \in \R^m \, : \, G(\alpha + \beta u) \leq h, \,\, G' u \leq h' \big\},
    \end{equation}
    with $\alpha = -K A_\mathrm{cl}^\inv \target_a$ and $\beta = I+K A_\mathrm{cl}^\inv B$, \change{where $I$ the identity matrix of appropriate size.}
    \end{itemize}    
\end{lemma}

\begin{proof}
    Item $i)$: We will show that the point $\tilde{x} = A_{\mathrm{cl}}^{-1}d_a \in \BackReachSt(\target_a, \controlSpace')$. 
    Note that this point $\tilde{x}$ is obtained for $\control' = 0$ in \eqref{eq:backward_reachable_set_stable}, which is an admissible input due to \cref{assumption:origin_in_barU}.
    Moreover, due to \cref{assumption:admissibility}, we have that $-K x \in \controlSpace$, and thus, it holds that $\tilde x \in \BackReachSt(\target_a, \controlSpace')$, which concludes the proof of item $i)$.    
    Item $ii)$: Solving the equality constraint in \eqref{eq:backward_reachable_set_stable} for $x$ yields $x = A_\mathrm{cl}^\inv (\target_a - Bu)$, so the input constraint $-Kx + u' \in \controlSpace$ can be written as
    \begin{equation}
        \label{eq:lemma1:proof1}
        -K A_\mathrm{cl}^\inv \target_a + (I + K A_\mathrm{cl}^\inv B) u' = \alpha + \beta u' \in \controlSpace.
    \end{equation}
    Thus, we have two convex polyhedral constraints on $u$, given by $\alpha + \beta u' \in \controlSpace = \{ u \in \R^m : Gu \leq h\}$ and $u' \in \controlSpace' = \{ u \in \R^m : G'u \leq h'\}$.
    The intersection of the feasible sets for $u$ is the set $\tilde\controlSpace$ in \eqref{eq:new_control_set}, concluding the proof of item $ii)$.
\end{proof}

Observe that, while we can compute $\BackReachSt(\target_a, \controlSpace')$ similarly as in \cref{sec:abstraction}, the number of vertices to consider is generally higher due to the additional input constraint $u' \in U'$. 

\subsection{Constructing smaller abstractions}
We can use the modified backward reachable set in \eqref{eq:backward_reachable_set_stable} to construct abstractions with fewer enabled actions, as illustrated by the following lemma. 

\begin{lemma}
    Consider the backward sets defined by \eqref{eq:backward_reachable_set} and \eqref{eq:backward_reachable_set_stable}. 
    If $\controlSpace' = \R^p$ in \eqref{eq:backward_reachable_set_stable}, then we have that $\BackReachSt(\target_a, \controlSpace') = \BackReach(\target_a)$.
    Moreover, for any two subsets $\controlSpace'' \subset \controlSpace' \subseteq \R^p$, it holds that $\BackReachSt(\target_a, \controlSpace'') \subseteq \BackReachSt(\target_a, \controlSpace')$.
    \label{lemma:intermediate-result}
\end{lemma}

\begin{proof}
    Letting $\controlSpace' = \R^p$ and $A_\mathrm{cl} = A-BK$ in \eqref{eq:backward_reachable_set_stable} gives
    \begin{alignat*}{2}
        \BackReachSt(\target_a, \controlSpace') &= 
        \big\{
            \state \in \R^n : \,\,
            &&\target_a = A\state + B(-K\state + \control'), \,\,
            \\
            & &&-K\state + \control' \in \controlSpace
        \big\}
        \\
        &= \big\{
            \state \in \R^n : \,\,
            &&\target_a = A\state + Bu, \,\, u \in \controlSpace
        \big\}
        \\
        &= \BackReach(\target_a), &&
    \end{alignat*}
    thus proving the first claim.
    For $\controlSpace'' \subset \controlSpace'$, observe from \eqref{eq:backward_reachable_set_stable} that $u' \in U'' \subset U'$.
    Thus, we obtain that $\BackReachSt(\target_a, \controlSpace'') \subseteq \BackReachSt(\target_a, \controlSpace')$, which proves the second claim.
\end{proof}

Recall from \cref{sec:abstraction} that an action $a \in \Actions$ is enabled in location $s \in S$ if and only if the corresponding partition element is contained in the backward reachable set $\BackReach(\target_a)$.
In the modified backward reachability set in \eqref{eq:backward_reachable_set_stable}, we can control the size of $\BackReachSt(\target_a)$ through $\controlSpace'$. 
In fact, \cref{lemma:intermediate-result} shows that by shrinking the set $\controlSpace'$, we may reduce the number of enabled actions at each state and, consequently, the size of the graph of the iMDP. 
In the next section, we show how suitable choices for the feedback gain $K$ and input constraint $\controlSpace'$ may lead to significantly smaller iMDP abstractions. 

%% file: sections/5-Experiments.tex
\section{Numerical Experiments}
\label{sec:experiments}

We implement our method in a Python tool, which is available at \url{https://github.com/LAVA-LAB/DynAbs}.
We use the model checker \prism~\cite{DBLP:conf/cav/KwiatkowskaNP11} to compute optimal policies as per \eqref{eq:optimal_policy} for iMDPs.
In all experiments, we apply \cref{thm:pac_control_design} with an overall confidence of $1-\beta \cdot |\Actions| \cdot |S| = 0.99$.
\change{For simplicity, we use partitions into rectangular regions.}

\subsection{Double integrator}
Consider a stochastic system with dynamics given as
{\small
\begin{equation*}
    \state_{k+1} = \begin{bmatrix}
        \frac{1}{\rho^2} & \frac{1 + \rho} + 1 \\
        0 & 1
    \end{bmatrix} \state_k
    + \begin{bmatrix}
        \frac{0.5 + \rho}{\rho} & 0.5 \\
        1 & 1
    \end{bmatrix} (-K\state_k + \control_k') + \noise_k,
    \label{eq:ExperimentDynamics:FullyActuated}
\end{equation*}}%
where we apply the control law given in \eqref{eq:two-layer-control-law}, and the noise $\eta_k \sim  \Gauss{0}{I_2}$ has a standard normal distribution and satisfies the conditions in Assumption \ref{assumption:unknown_noise}.
We select $\rho = 2$ to render the system unstable when a trivial control of $-Kx_k + u'_k = 0$ is applied. 
The reach-avoid task is to reach a state $x \in [-3, 3]^2$ while avoiding states $x \notin [-41,41]^2$ within $\finiteHorizon = 16$ steps.\footnote{\change{The correctness of our iMDP abstraction is independent of the horizon. For numerical experiments with an infinite time horizon, we refer to~\cite{DBLP:journals/jair/BadingsRAPPSJ23}.}}
We partition the set $\partitionSpace = [-41,41]^2$ into $41$ by $41$ rectangular regions of width two.
The input constraint is $\controlSpace = [-60,60]^2$.

\pgfmathsetlengthmacro\MajorTickLength{
      \pgfkeysvalueof{/pgfplots/major tick length} * 0.5
    }

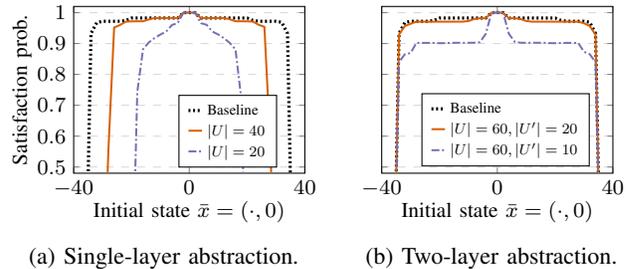
\begin{figure}[t]
\begin{subfigure}[b]{0.49\linewidth}
    \centering
    \input{figures/results_baseline}    
    \caption{Single-layer abstraction.}
    \label{fig:integrator_baseline}
\end{subfigure}
\begin{subfigure}[b]{0.49\linewidth}
    \centering
    \input{figures/results_stabilized}
    \caption{Two-layer abstraction.}
    \label{fig:integrator_stabilized}
\end{subfigure}
    \caption{Lower bound satisfaction probabilities (obtained from \cref{thm:pac_control_design}) for the integrator experiment with initial conditions $\bar\state = (x_1, 0)$ for all $-41 \leq x_1 \leq 41$.}
    \label{fig:integrator}
\end{figure}

\newcommand\setrow[1]{\gdef\rowmac{#1}#1\ignorespaces}
\newcommand\clearrow{\global\let\rowmac\relax}
\clearrow

\begin{table}[t!]
\centering

\caption{Number of iMDP transitions for the integrator experiment. The highlighted rows are those shown in \cref{fig:integrator}.}
\begin{tabular}{lllHr}
    \toprule
    Stabilized? & $\controlSpace$ & $\controlSpace'$ & Locations & iMDP transitions
    \\
    \midrule	
    \belowrulesepcolor{LightCyan}\rowcolor{LightCyan}
    No (baseline) & $[-60,60]^2$ & n.a. & $1\,684$ & $39\,773\,745$ \\
    \rowcolor{LightCyan}
    No & $[-40,40]^2$ & n.a. & $1\,684$ & $21\,289\,058$ \\
    \rowcolor{LightCyan}
    No & $[-20,20]^2$ & n.a. & $1\,684$ & $5\,219\,518$ \\
    \aboverulesepcolor{LightCyan} 
    \midrule
    Yes & $[-60,60]^2$ & $[-30,30]^2$ & $1\,684$ & $25\,671\,576$ \\
    \rowcolor{LightCyan}
    Yes & $[-60,60]^2$ & $[-20,20]^2$ & $1\,684$ & $15\,757\,546$ \\
    \rowcolor{LightCyan}
    Yes & $[-60,60]^2$ & $[-10,10]^2$ & $1\,684$ & $3\,996\,029$ \\
    Yes & $[-40,40]^2$ & $[-20,20]^2$ & $1\,684$ & $11\,691\,267$ \\
    Yes & $[-40,40]^2$ & $[-10,10]^2$ & $1\,684$ & $2\,895\,878$ \\
    Yes & $[-20,20]^2$ & $[-10,10]^2$ & $1\,684$ & $1\,034\,996$ \\
\bottomrule
\end{tabular}
\label{table:integrator}
\end{table}

\subsubsection*{Baseline}
As a baseline, we set $K=0$ and construct the single-layer iMDP abstraction (as outlined in \cref{sec:abstraction}) with input constraint $\controlSpace = [-60,60]^2$.
The resulting iMDP has 39.8 million transitions, and the lower bounds on the satisfaction probabilities (obtained from \cref{thm:pac_control_design}) are shown in \cref{fig:integrator} for a range of initial states $\bar{x} = (x_1, 0)$ for $x_1 \in [-41, 41]$.

\subsubsection*{Stabilizing controller}
We now use our two-layer abstraction scheme, where we compute the gain $K$ with an LQR with cost matrices $Q = R = I_2$, yielding $(A-BK)$ having eigenvalues of $\lambda = 0.178 \pm 0.136i$.
We construct the iMDP for the different sets $\controlSpace$ and $\controlSpace'$ shown in \cref{table:integrator}, presenting the lower bound satisfaction probabilities for two cases in \cref{fig:integrator}.
With our method, we construct significantly smaller abstractions with little loss in probabilistic guarantee.
For example, with $|\controlSpace| = 60$, $|\controlSpace'| = 20$, the number of iMDP transitions is reduced from 39.8 to 15.8 million at negligible loss in probabilistic guarantee.
Shrinking the set $\controlSpace'$ further reduces the iMDP size; however, at the cost of a considerable reduction in probabilistic guarantee, as shown in \cref{fig:integrator_baseline}.

\begin{figure}[t]
\begin{subfigure}[b]{0.495\linewidth}
    \centering
    \includegraphics[width=\linewidth]{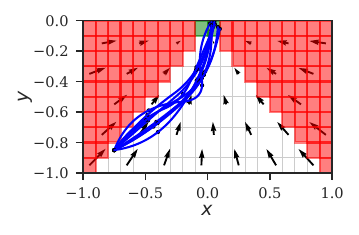}
    \caption{Aligned property.}
    \label{fig:spacecraft_aligned}
\end{subfigure}
\begin{subfigure}[b]{0.495\linewidth}
    \centering
    \includegraphics[width=\linewidth]{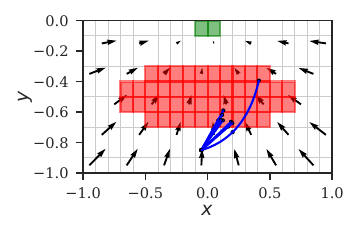}
    \caption{Disaligned property.}
    \label{fig:spacecraft_disaligned}
\end{subfigure}
    \caption{Simulated trajectories and stabilized vector fields $(A-BK)x$ for both reach-avoid properties considered in the spacecraft problem (goal states in green; unsafe states in red). Case (a) is aligned, while case (b) is not.}
    \label{fig:spacecraft}
\end{figure}

\subsection{Spacecraft docking}
We consider the spacecraft docking problem from~\cite{DBLP:conf/hybrid/VinodGO19}, with $x \in \R^4$ modeling the position and velocity in two dimensions 
(see~\cite{DBLP:conf/hybrid/VinodGO19} for the full model).
We illustrate that our method generally works well if the stabilizing feedback controller is aligned with the reach-avoid property.
That is, the stabilizing controller should \emph{steer} the state $x$ towards the goal region $\goalRegion \subset \R^4$, while steering clear from the unsafe states $\unsafeRegion \subset \R^4$.
We consider the two reach-avoid problems shown in \cref{fig:spacecraft} (only the position state variables are shown).
As a baseline, we construct the abstraction with a partition into $3\,200$ elements and an input constraint $\controlSpace = [-0.1, 0.1]^2$.
For the reach-avoid problem in \cref{fig:spacecraft_aligned}, the resulting iMDP has 1.6 million transitions and leads to a lower bound satisfaction probability of $0.80$ in \cref{thm:pac_control_design}.
Similarly, for the problem in \cref{fig:spacecraft_disaligned}, the iMDP has 2.1 million transitions and leads to a lower bound satisfaction probability of $0.86$.

Now, we apply our method with $\controlSpace' = [-0.08, 0.08]^2$, resulting in iMDPs with 280 and 330 thousand transitions (reductions of 79\% and 85\% respectively).
For the problem in \cref{fig:spacecraft_aligned}, the lower bound on the satisfaction probability is $0.79$ (only $0.01$ below the baseline).
However, for \cref{fig:spacecraft_disaligned}, the bound drops to $0.0072$, i.e., almost zero.
To explain this severe performance loss, we show simulated trajectories under the refined controller (as per \cref{def:refined_controller}) in \cref{fig:spacecraft}.
Moreover, the arrows show the vector field under the stabilized dynamics, i.e., $(A-BK)x$ for different $x \in \R^n$.
In \cref{fig:spacecraft_aligned}, the vector field points to the goal region and away from unsafe states, and is thus aligned with the property.
By contrast, the vector field in \cref{fig:spacecraft_disaligned} is not aligned since it steers the system into unsafe states, causing a performance loss of the controller.

%% file: figures/results_baseline.tex
\footnotesize\begin{tikzpicture}[]
  \begin{axis}[
      width=1.1\linewidth,
      height=3.8cm,
      ymajorgrids,
      grid style={dashed,gray!30},
      xlabel={Initial state $\bar{x}=(\cdot, 0)$},
      ylabel={Satisfaction prob.},
      x label style={at={(axis description cs:0.5,-0.12)}},
      y label style={at={(axis description cs:-0.15,0.5)}},
      xmin=-40,
      xmax=40,
      ymin=0.48,
      ymax=1.02,
      ytick={0.5, 0.6, 0.7, 0.8, 0.9, 1.0},
      xtick={-40, 0, 40},
      major tick length=\MajorTickLength,
      cycle list/Dark2,
      every axis plot/.append style={line width=0.75pt},
      legend cell align={left},
      legend columns=1,
      legend image post style={xscale=0.3},
      legend style={at={(0.88,0.47)},
                    nodes={scale=0.7, transform shape},
                    anchor=north east, 
                    column sep=0ex,}
    ]
    
    \addplot [black, densely dotted, very thick] table[x=x, y=y] {figures/data3/cross_section_stab=False_uMax=60.0.csv};
    \addplot [index of colormap=1 of Dark2] table[x=x, y=y] {figures/data3/cross_section_stab=False_uMax=40.0.csv};
    \addplot [index of colormap=2 of Dark2, densely dashdotted] table[x=x, y=y] {figures/data3/cross_section_stab=False_uMax=20.0.csv};

    \legend{
        {Baseline}, 
        {$|\controlSpace| = 40$},
        {$|\controlSpace| = 20$},
    }
    
  \end{axis}
\end{tikzpicture}%

%% file: figures/results_stabilized.tex
\footnotesize\begin{tikzpicture}[]
  \begin{axis}[
      width=1.1\linewidth,
      height=3.8cm,
      ymajorgrids,
      grid style={dashed,gray!30},
      xlabel={Initial state $\bar{x}=(\cdot, 0)$},
      ylabel=\empty,
      x label style={at={(axis description cs:0.5,-0.12)}},
      xmin=-40,
      xmax=40,
      ymin=0.48,
      ymax=1.02,
      ytick={0.5, 0.6, 0.7, 0.8, 0.9, 1.0},
      yticklabels=\empty,
      xtick={-40, 0, 40},
      major tick length=\MajorTickLength,
      cycle list/Dark2,
      every axis plot/.append style={line width=0.75pt},
      legend cell align={left},
      legend columns=1,
      legend image post style={xscale=0.3},
      legend style={at={(0.90,0.48)},
                    nodes={scale=0.7, transform shape},
                    anchor=north east, 
                    column sep=0ex,}
    ]

    \addplot [black, densely dotted, very thick] table[x=x, y=y] {figures/data3/cross_section_stab=False_uMax=60.0.csv};
    \addplot [index of colormap=1 of Dark2] table[x=x, y=y] {figures/data3/cross_section_stab=True_uMax=60.0_uBarMax=20.0.csv};
    \addplot [index of colormap=2 of Dark2, densely dashdotted] table[x=x, y=y] {figures/data3/cross_section_stab=True_uMax=60.0_uBarMax=10.0.csv};


    \legend{
        {Baseline},
        {$|\controlSpace| = 60, |\controlSpace'| = 20$},
        {$|\controlSpace| = 60, |\controlSpace'| = 10$}
    }
    
  \end{axis}
\end{tikzpicture}%

%% file: sections/6-Conclusion.tex
\section{Conclusion}
\label{sec:conclusion}

In this paper, we have developed a novel formal abstraction method for stochastic linear dynamical systems that exploits stability to generate smaller abstract models.
By stabilizing the dynamics with a linear feedback gain first, we have shown that we can reduce the size of abstractions (in terms of the number of edges in the underlying graph) significantly.
Our experiments have shown that, when the feedback gain steers the system toward the goal states (and away from the unsafe states), we can reduce the number of transitions by up to $90\%$ with negligible performance loss.

However, if the stabilizing controller is not aligned with the control task (as in \cref{fig:spacecraft_disaligned}), the controller performance degrades significantly.
One possible solution is to use a piece-wise affine trajectory-tracking controller, which selects different gains in different regions of the state space. 
Exploring this latter approach will be the next step of our research.